\documentclass[a4paper,11pt]{article}
\pdfoutput=1

\usepackage[margin=1in]{geometry}
\usepackage{amssymb,amsmath,amsthm}
\usepackage{times}
\usepackage[T1]{fontenc}
\usepackage[ruled,vlined,linesnumbered,lined]{algorithm2e}
\usepackage{color,tikz,xcolor,pgf,pgffor,rotating,pgfplots}
\usepackage{tikz}
\tikzstyle{nloc}=[draw, text badly centered, rectangle, rounded corners, minimum size=2em,inner sep=0.5em]
\tikzstyle{loc}=[draw,rectangle,minimum size=1.4em,inner sep=0em]
\tikzstyle{trans}=[-latex, rounded corners]
\tikzstyle{trans2}=[-latex, dashed, rounded corners]

\usetikzlibrary{fit} 
\usetikzlibrary{backgrounds} 
\usepgflibrary{shapes}
\usetikzlibrary{arrows}
\usetikzlibrary{decorations}
\usetikzlibrary{automata}
\tikzstyle{background}=[rectangle,fill=gray!10, inner sep=0.1cm, rounded corners=0mm]

\theoremstyle{plain}
\newtheorem{theorem}{Theorem}[section]
\newtheorem{lemma}[theorem]{Lemma}

\theoremstyle{definition}
\newtheorem{definition}[theorem]{Definition} 
\newtheorem{example}[theorem]{Example}

\def\rmdef{\stackrel{\mbox{\rm {\tiny def}}}{=}} 

\newcommand \Aa {\mathcal{A}}
\newcommand \Gg {\mathcal{G}}
\newcommand \Dd {\mathbb{D}}
\newcommand \PZero {Player Max}
\newcommand \POne  {Player Min}
\newcommand \Pmax {Player Max}
\newcommand \Pmin  {Player Min}
\newcommand \mMAX {\textrm{Max}}
\newcommand \mMIN  {\textrm{Min}}

\newcommand{\seq}[1]{\langle #1 \rangle}

\newcommand{\Plays}{\mathsf{Plays}}
\newcommand{\Real}{\mathbb{R}}

\newcommand\prof{\mathsf{Prof}}
\newcommand\val{\mathsf{val}}

\begin{document}

\title{Symmetric Strategy Improvement}
\author{Sven Schewe\textsuperscript{1}, Ashutosh Trivedi\textsuperscript{2}, and Thomas Varghese\textsuperscript{1} \\ }
\date{\emph{\textsuperscript{1}Department of Computer Science, University of Liverpool \\ \medskip \textsuperscript{2}Department of Computer Science and Engineering\\ Indian Institute of Technology -- Bombay}}

\pagestyle{plain}

\bibliographystyle{plain}

\maketitle

\begin{abstract}
Symmetry is inherent in the definition of most of the two-player zero-sum games,
including parity, mean-payoff, and discounted-payoff games. 
It is therefore quite surprising that no symmetric analysis techniques for these
games exist.
We develop a novel symmetric strategy improvement algorithm where, in each
iteration, the strategies of both players are improved simultaneously.   
We show that symmetric strategy improvement defies Friedmann's traps, which
shook the belief in the potential of classic strategy improvement to be
polynomial.  
\end{abstract}

\section{Introduction}
\label{sec:intro}
We study turn-based graph games between two
players---Player Min and Player Max---who take turns to move a token along the
vertices of a coloured finite graph so as to optimise their adversarial
objectives. 
Various classes of graph games are characterised by the objective of the
players, for instance in \emph{parity games} the objective is to optimise the parity
of the dominating colour occurring infinitely often, while in \emph{discounted and
  mean-payoff games} the objective is the discounted and limit-average sum of
the colours.  
Solving graph games is the central and most expensive step
in many model
checking~\cite{Kozen/83/mu,Emerson+all/93/mu,Wilke/01/Alternating,deAlfaro+Henziger+Majumdar/01/Control,Alur+Henziger+Kupferman/02/ATL},
satisfiability
checking~\cite{Wilke/01/Alternating,Kozen/83/mu,Vardi/98/2WayAutomata,Schewe+Finkbeiner/06/ATM},
and synthesis~\cite{Piterman/06/Parity,Schewe+Finkbeiner/06/Asynchronous}
algorithms.
More efficient algorithms for solving graph games will therefore foster the
development of performant model checkers and contribute to bringing synthesis
techniques to practice.
 
Parity games enjoy a special status among graph games and the quest for performant
algorithms \cite{Kozen/83/mu,Emerson+Lei/86/Parity,Emerson+Jutla/91/Memoryless,McNaughton/93/Games,Zwick+Paterson/96/payoff,Browne-all/97/fixedpoint,Zielonka/98/Parity,Jurdzinski/00/ParityGames,Ludwig/95/random,Puri/95/simprove,Voge+Jurdzinski/00/simprove,BjorklundVorobyov/07/subexp,Obdrzalek/03/TreeWidth,Lange/05/ParitySAT,Berwanger+all/06/ParityDAG,Jurdzinski/06/subex,Schewe/07/parity,Schewe/08/improvement,Fearnley/10/snare}  
for solving them has therefore been an active field of research during the last
decades. 
Traditional forward techniques \mbox{($\approx
O(n^{\frac{1}{2}c})$~\cite{Jurdzinski/00/ParityGames}} for parity games with $n$
positions 
and $c$ colours), backward techniques ($\approx\hspace*{-1.9pt}
O(n^{c})$~\cite{McNaughton/93/Games,Emerson+Lei/86/Parity,Zielonka/98/Parity}),
and their combination ($\approx \hspace*{-1.9pt}
O(n^{\frac{1}{3}c})$~\cite{Schewe/07/parity}) provide good complexity bounds. 
However, these bounds are sharp, and techniques with good complexity
bounds~\cite{Schewe/07/parity,Jurdzinski/00/ParityGames} frequently display
their worst case complexity on practical examples. 
Strategy improvement
algorithms~\cite{Ludwig/95/random,Puri/95/simprove,Voge+Jurdzinski/00/simprove,BjorklundVorobyov/07/subexp,Schewe/08/improvement,Fearnley/10/snare},
on the other hand, are closely related to the Simplex algorithm for solving
linear programming problems that perform well in practice. 

Classic strategy improvement algorithms are built around 
the existence of optimal positional strategies for both players.  
They start with an arbitrary positional strategy for a player and
iteratively compute a better positional strategy in every step until the
strategy cannot be further improved. 
Since there are only finitely many positional strategies in a finite graph, termination is guaranteed. 
The crucial step in a strategy improvement algorithm is to compute a better
strategy from the current strategy. 
Given a current strategy $\sigma$ of a player (say, \Pmax), this step is performed
by first computing the globally optimal counter 
strategy $\tau^c_\sigma$ of the opponent (\Pmin) and then computing the
value of each vertex of the game restricted to the strategies $\sigma$ and
$\tau^c_\sigma$.
For the games under discussion (parity, discounted, and mean-payoff) both of
these computations are simple and tractable.
This value dictates potentially locally profitable changes or switches
$\prof(\sigma)$ that \Pmax\  can make vis-{\`a}-vis his previous strategy $\sigma$. 
For the correctness of the strategy improvement algorithm it is required that such
locally profitable changes imply a global improvement. 
The strategy of \Pmax\ can then be updated according to a switching rule
(akin to pivoting rule of the Simplex) in order to give an improved strategy. 
This has led to the following template for classic strategy improvement algorithms.
\begin{algorithm}
  \caption{\label{alg:classic-sia} Classic strategy improvement algorithm}
  determine an optimal counter strategy $\tau^c_\sigma$ for $\sigma$\\ 
  evaluate the game for $\sigma$ and $\tau^c_\sigma$ and determine the
  profitable changes $\prof(\sigma)$ for $\sigma$ \\
  update $\sigma$ by applying changes from $\prof(\sigma)$ to $\sigma$\\
\end{algorithm}

A number of switching rules, including the ones inspired by Simplex pivoting
rules, have been suggested for strategy improvement algorithms.
The most widespread ones are to select changes for all game states where this is
possible, choosing a combination of those with an optimal update guarantee, or
to choose uniformly at random. 
For some classes of games, it is also possible to select an optimal combination
of updates \cite{Schewe/08/improvement}. 
There have also been suggestions to use more advanced randomisation techniques
with sub-exponential -- $2^{O(\sqrt{n})}$ -- bounds
\cite{BjorklundVorobyov/07/subexp} and snare memory \cite{Fearnley/10/snare}. 
Unfortunately, all of these techniques have been shown to be exponential in the
size of the game
\cite{Friedmann/11/lower,Friedmann/11/Zadeh,Friedmann/13/snare}. 

Classic strategy improvement algorithms treat the two players involved quite
differently where at each iteration one player computes a globally optimal
counter strategy, while the other player performs local updates.
In contrast, a  \emph{symmetric strategy improvement} algorithm symmetrically 
improves the strategies of both players at the same time, and uses the finding
to guide the strategy improvement.
This suggests the following na\"ive symmetric approach.
\begin{algorithm}
  \caption{\label{alg:classic-ssia} Na\"ive symmetric strategy improvement algorithm}
  determine $\tau' = \tau^c_\sigma$ \hspace{50mm} determine $\sigma' = \sigma^c_\tau$\\
  update $\sigma$ to $\sigma'$  \hspace{56.6mm} update $\tau$ to $\tau'$\\
\end{algorithm}

This algorithm has earlier been suggested by Condon~\cite{Condon93onalgorithms}
where it was shown that a repeated application of this update can lead to cycles
\cite{Condon93onalgorithms}.   
A problem with this na\"ive approach is that there is no guarantee that the
primed strategies are generally better than the unprimed ones. 
With hindsight this is maybe not very surprising, as in particular no improvement
in the evaluation of running the game with $\sigma',\tau'$ can be expected over
running the game with $\sigma,\tau$, as an improvement for one player is on the
expense of the other. 
This observation led to the approach being abandoned.
In this paper we propose the following  more careful symmetric strategy
improvement algorithm that guarantees improvements in each iteration similar to
classic strategy improvement.  

\begin{algorithm}
  \caption{\label{alg:novel-ssia} Symmetric strategy improvement algorithm}
  determine $\tau^c_\sigma$
  \hspace{58mm}
  determine $\sigma^c_\tau$\\
  determine $\prof(\sigma)$ for $\sigma$
  \hspace{40.3mm} 
  determine $\prof(\tau)$ for $\tau$\\
  update $\sigma$ using $\prof(\sigma) \cap \sigma^c_\tau$
  \hspace{33.4mm} 
  update $\tau$ using $\prof(\tau) \cap \tau^c_\sigma$\\
\end{algorithm}

The main difference to classic strategy improvement approaches is that we exploit the
strategy of the other player to inform the search for a good improvement step. 
In this algorithm we select only such updates to the two strategies that agree
with the optimal counter strategy to the respective other's strategy. 
We believe that this will provide a gradually improving advice function that
will lead to few iterations. 
We support this assumption by showing that this algorithm suffices to escape the
traps Friedmann has laid to establish lower bounds for different types of
strategy improvement algorithms
\cite{Friedmann/11/lower,Friedmann/11/Zadeh,Friedmann/13/snare}.

\section{Preliminaries}
\label{sec:prelims}

We focus on turn-based zero-sum games played
between two players---\PZero{} and \POne{}---over finite graphs.  
A game arena $\Aa$ is a tuple $(V_\mMAX, V_\mMIN, E, C, \phi)$ where
$(V = V_\mMAX \cup V_\mMIN, E)$ is a finite directed graph with the set of
vertices $V$ partitioned into a set $V_\mMAX$ of vertices controlled by \Pmax\
and a set $V_\mMIN$ of vertices controlled by \Pmin, 
$E \subseteq V \times V$ is the set of edges,  $C$ is a set of colours, 
$\phi: V \to C$ is the colour mapping. 
We require that every vertex has at least one outgoing edge. 

A turn-based game over $\Aa$ is played between players by moving a token along
the edges of the arena. 
A play of such a game starts by placing a token on some initial vertex 
$v_0 \in V$.
The player controlling this vertex then chooses a successor vertex
$v_1$ such that $(v_0, v_1) \in E$ and the token is moved to this successor vertex. 
In the next turn the player controlling the vertex $v_1$ chooses the successor
vertex $v_2$ with $(v_1, v_2) \in E$ and the token is moved accordingly. 
Both players move the token over the arena in this manner and thus form a play
of the game. 
Formally, a play of a game over $\Aa$ is an infinite sequence of vertices
$\seq{v_0, v_1, \ldots} \in V^\omega$ such that, for all $i \geq 0$, we have that
$(v_i, v_{i+1}) \in E$.  
We write $\Plays_\Aa(v)$ for the set of plays over $\Aa$ starting from vertex
$v \in V$ and $\Plays_\Aa$ for the set of plays of the game. 
We omit the subscript when the arena is clear from the context. 
We extend the colour mapping $\phi: V \to C$ from vertices to plays by defining
the mapping $\phi: \Plays \to C^\omega$ as  
$\seq{v_0, v_1, \ldots} \mapsto \seq{\phi(v_0), \phi(v_1), \ldots}$.

\begin{definition}[Graph Games]
  A graph game $\Gg$ is a tuple $(\Aa, \eta, \prec)$ such that  $\Aa$ is
  an \emph  {arena}, $\eta: C^\omega \to \Dd$ is an evaluation function where
  $\Dd$ is the carrier set of a complete space, and $\prec$ is a 
  preference ordering over $\Dd$. 
\end{definition}

\begin{example}
  Parity, mean-payoff and discounted payoff games are graph games 
  $(\Aa, \eta, \prec)$ played on game arenas 
  $\Aa  = (V_\mMAX, V_\mMIN, E, \Real, \phi)$.
  For mean payoff games the evaluation function is $\eta : \seq{c_0, c_1, \ldots}
  \mapsto \liminf_{i \rightarrow \infty} \frac{1}{i}\sum_{j=0}^{i-1} c_j$, while
  for discounted payoff games with discount factor $\lambda \in [0,1)$ it is
    $\eta : \seq{c_0, c_1, \ldots} \mapsto  \sum_{i=0}^{\infty} \lambda^i c_i$
    with $\prec$ as the natural order over the reals.  
    For (max) parity games the evaluation function is 
    $\eta : \seq{c_0, c_1, \ldots} \mapsto \limsup_{i \rightarrow \infty} c_i$ 
    often used with a preference order $\prec_{\mathrm{parity}}$ where higher
    even colours are preferred over smaller even colours, even colours are
    preferred over odd colours, and smaller odd colours are preferred over higher 
    odd colours. 

In the remainder of this paper, we will use parity games where every colour is
unique, i.e., where $\phi$ is injective. 
All parity games can be translated into such games as discussed in
\cite{Voge+Jurdzinski/00/simprove}. 
For these games, we use a valuation function based on their progress measure. 
We define $\eta$ as $\seq{c_0, c_1, \ldots} \mapsto (c, C, d)$, where $c = \limsup_{i
  \rightarrow \infty} c_i$ is the dominant colour of the colour sequence, $d =
\min\{ i \in \omega \mid c_i = c\}$ is the index of the first occurrence of
$c$, and $C = \{c_i \mid i < d, c_i > c\}$ is the set of colours that occur
before the first occurrence of $c$. 
The preference order is defined as the following: we have $(c', C', d') \prec
(c, C, d)$ if
\begin{itemize} 
\item $c' \prec_{\mathrm{parity}} c$,
\item $c {=} c'$, the highest colour $h$ in the symmetric difference between $C$ and $C'$ is even,~and~in~$C$,
  
\item $c {=} c'$, the highest colour $h$ in the symmetric difference between $C$ and $C'$ is odd, and~in~$C'$,
  
\item $c = c'$ is even, $C = C'$, and $d<d'$, or
  
\item $c = c'$ is odd, $C = C'$, and $d>d'$. 
\end{itemize}
\end{example}


\begin{definition}[Strategies]
  A strategy of \Pmax\ is a function $\sigma: V^*V_\mMAX \rightarrow V$ such that
  $\big(v,\sigma(\pi v)\big) \in E$ for all $\pi \in V^*$ and $v \in V_\mMAX$. 
  Similarly, a strategy of \Pmin\ is a function $\tau: V^*V_\mMIN \rightarrow V$
  such that $\big(v,\sigma(\pi v)\big) \in E$ for all $\pi \in V^*$ and $v \in
  V_\mMIN$.  
  We write $\Sigma^\infty$ and $T^\infty$ for the set of strategies of
  \Pmax\ and \Pmin, respectively.
\end{definition}

\begin{definition}[Valuation]
 For a strategy pair  $(\sigma, \tau) \in \Sigma^\infty \times T^\infty$
and an initial vertex $v \in V$ we denote the unique play starting from the
vertex $v$ by $\pi(v, \sigma, \tau)$ and we write $\val_\Gg(v, \sigma, \tau)$ for
the value of the vertex $v$ under the strategy pair $(\sigma, \tau)$ defined as 
\[
\val_\Gg(v, \sigma, \tau) \rmdef \eta\big(\phi(\pi(v, \sigma, \tau))\big).
\] 

We also define the concept of the value of a strategy $\sigma \in \Sigma^\infty$
and $\tau \in T^\infty$ as 
\[
\val_{\Gg}(v,\sigma) \rmdef  \inf_{\tau \in T^\infty} \val_{\Gg}(v,\sigma,\tau)
\text{ and } 
\val_{\Gg}(v,\tau) \rmdef \sup_{\sigma \in \Sigma^\infty}
\val_\Gg(v,\sigma,\tau).
\]
We also extend the valuation for vertices to a valuation for the whole game by
defining $V$ dimensional vectors $\val_{\mathcal G}(\sigma) : v \mapsto \val_{\mathcal
  G}(v,\sigma)$ with the usual $V$ dimensional partial order $\sqsubseteq$,
where $\val \sqsubseteq \val'$ if, and only if, $\val(v) \preceq \val'(v)$ holds
for all $v \in V$. 
\end{definition}

\begin{definition}[Positional Determinacy]
  We say that a strategy $\sigma \in \Sigma^\infty$  is memoryless or
  \emph{positional} if it only depends on the last state, i.e. for all 
  $\pi, \pi' \in V^*$ and $v \in V_\mMAX$ we have that $\sigma(\pi v) = \sigma(\pi' v)$.  
  Thus, a positional strategy can be viewed as a function $\sigma: V_\mMAX \to V$
  such that for all $v\in V_\mMAX$ we have that $(v, \sigma(v)) \in E$.
  The concept of positional strategies of \Pmin\ is defined in an analogous
  manner. 
  We write $\Sigma$ and $T$ for the set of positional strategies of Players Max
  and Min, respectively. 
  We say that a game is positionally determined if: 
  \begin{itemize}
  \item $\val_{\mathcal G}(v,\sigma) = \min_{\tau \in T} \val_{\mathcal
    G}(v,\sigma,\tau)$ holds for all $\sigma \in \Sigma$, 
    
  \item $\val_{\mathcal G}(v,\tau) = \max_{\sigma \in \Sigma} \val_{\mathcal
    G}(v,\sigma,\tau)$ holds for all $\tau \in T$, 
    
  \item {\bf Existence of value}: for all $v \in V$ $\max_{\sigma \in \Sigma} \val_{\mathcal G}(v,\sigma) =
    \min_{\tau \in T} \val_{\mathcal G}(v,\tau)$ holds, and we use $\val_{\mathcal
      G}(v)$ to denote this value, and 
    
  \item {\bf Existence of positional optimal strategies}:  there is a pair $\tau_{\min},\sigma_{\max}$ of strategies such that, for
    all $v \in V$, $\val_{\mathcal G}(v) = \val_{\mathcal G}(v,\sigma_{\max}) =
    \val_{\mathcal G}(v,\tau_{\min})$ holds. 
    Observe that for all $\sigma \in \Sigma$ and $\tau \in T$ we have that 
    $\val_{\mathcal
        G}(\sigma_{\max}) \sqsupseteq \val_{\mathcal G}(\sigma)$
    and $\val_{\mathcal G}(\tau_{\min}) \sqsubseteq \val_{\mathcal G}(\tau)$.
  \end{itemize}
\end{definition}

Observe that (first and second item above) that classes of games with positional
strategies guarantee an optimal positional counter strategy for \Pmin\ to all
strategies $\sigma \in \Sigma$ of \Pmax. 
We denote these strategies by $\tau^c_\sigma$. 
Similarly, we denote the optimal positional counter strategy for \Pmax\ to a
strategy $\tau \in T$ by $\sigma^c_\tau$ of \Pmin. 
While this counter strategy is not necessarily unique, we use the {\bf
  convention} in all proofs that $\tau^c_\sigma$ is always the same counter
strategy for $\sigma \in \Sigma$, and $\sigma^c_\tau$ is always the same counter
strategy for $\tau \in T$.  
 
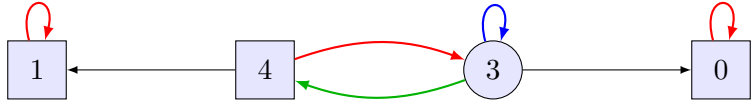
\begin{figure}[t]
  \begin{center}
    \scalebox{1}{
      \begin{tikzpicture}
        \tikzstyle{eloc}=[fill=blue!10!white,draw,circle,minimum size=2em,inner sep=0em]
        \tikzstyle{oloc}=[fill=blue!10!white,draw,minimum size=2em,inner sep=0em]
        
        \tikzstyle{trans}=[-latex, rounded corners]

        \node[oloc] at (0,0) (1)  {$1$};
        \node[oloc] at (3,0) (4)  {$4$};
        \node[eloc] at (6,0) (3)  {$3$};
        \node[oloc] at (9,0) (0)  {$0$};
        
        \draw[trans] (4) -- (1);
        \draw[trans] (3) -- (0);
        
        \draw[trans,color=red,thick] (4) to [bend left=20] (3);
        \draw[trans,color=green!70!black,thick] (3) to [bend left=20] (4);
        
        \draw[trans,color=red,thick] (1) to [loop above]  (1);
        \draw[trans,color=blue,thick] (3) to [loop above]  (3);
        \draw[trans,color=red,thick] (0) to [loop above]  (0);
      \end{tikzpicture}
    }
  \end{center}
  \caption{Parity game arena with four vertices and unique colours.}
  \vspace{-1em}
  \label{fig:example}
\end{figure}

\begin{example}
 Consider the parity game arena shown in Figure~\ref{fig:example}.
  We use circles for the vertices of Player Max and squares for Player Min. 
  We label each vertex with its colour.
  Notice that a positional strategy can be depicted just by specifying an
  outgoing edge for all the vertices of a player.
  The positional strategies $\sigma$ of \PZero{} is depicted in {\color{blue}
    blue} and the positional strategy $\tau$ of \POne{} is depicted in
  {\color{red} red}. 
  In the example, $\val(1,\sigma,\tau) = (1,\emptyset,0)$, $\val(4,\sigma,\tau) =
  (3,\{4\},1)$, $\val(3,\sigma,\tau) = (3,\emptyset,0)$, and $\val(0,\sigma,\tau)
  = (0,\emptyset,0)$. 
\end{example}

\subsection{Classic Strategy Improvement Algorithm}
As discussed in the introduction, classic strategy improvement algorithms work
well for classes of games that are positionally determined. 
Moreover, the evaluation function should be such that one can easily identify
the set $\prof(\sigma)$ of profitable updates and reach an optimum exactly where
there are no profitable updates.  
We formalise these prerequisites for a class of games to be good for strategy
improvement algorithm in this section. 




\begin{definition}[Profitable Updates]
  For a strategy $\sigma \in \Sigma$, an edge $(v, v') \in E$ with $v \in V_\mMAX$
  is a profitable update if $\sigma' \in \Sigma$ with $\sigma': v \mapsto v'$ and
  $\sigma': v'' \mapsto \sigma(v'')$ for all $v'' \neq v$ has a strictly greater
  evaluation than $\sigma$, $\val_{\mathcal G}(\sigma') \sqsupset \val_{\mathcal
    G}(\sigma)$. 
  We write $\prof(\sigma)$ for the set of profitable updates. 
\end{definition} 
\begin{example}
In our example from Figure \ref{fig:example}, $\tau = \tau_\sigma^c$ is the
optimal counter strategy to $\sigma$, such that $\val(\sigma) =
\val(\sigma,\tau)$. 
$\prof(\sigma) = \{(3,4),(3,0)\}$, because both the successor to the left and
the successor to the right have a better valuation, $(3,\{4\},1)$ and
$(0,\emptyset,0)$, respectively, than the successor on the selected self-loop,
$(3,\emptyset,0)$. 
\end{example}

For a strategy $\sigma$ and a functional (right-unique) subsets 
$P \subseteq \prof(\sigma)$ we define the strategy $\sigma^P$ with $\sigma^P: v
\mapsto v'$ if $(v,v') \in P$ and $\sigma^P: v \mapsto \sigma(v)$ if there is no
$v' \in V$ with $(v,v') \in P$.
For a class of graph games, profitable updates are \emph{combinable} if, for all
strategies $\sigma$ and all functional (right-unique) subsets $P \subseteq
\prof(\sigma)$ we have that 
$\val_{\mathcal G}(\sigma^P) \sqsupset \val_{\mathcal G}(\sigma)$. 
Moreover, we say that a class of graph games is \emph{maximum identifying} if
$\prof(\sigma) = \emptyset  \Leftrightarrow \val_{\mathcal G}(\sigma) =
\val_{\mathcal G}$.
Algorithm~\ref{alg:classic-sia-conc} provides a generic template for strategy
improvement algorithms. 
\begin{algorithm}[h]
  \caption{\label{alg:classic-sia-conc} Classic strategy improvement algorithm}
  Let $\sigma_0$ be an arbitrary positional strategy. {\bf Set} $i :=0$.\\
  If $\prof(\sigma_i) = \emptyset$ \Return $\sigma_i$\\ 
  $\sigma_{i+1} := {\sigma_i}^P$ for some functional subset $P\subseteq
  \prof(\sigma)$. {\bf Set} $i := i + 1$. {\bf go to} 2.\\
\end{algorithm}

We say that a class of games is \emph{good for $\max$ strategy improvement} if
they are positionally determined and have combinable and maximum identifying
improvements. 
\begin{theorem}
\label{theorem:classic}
If a class of games is good for $\max$ strategy improvement then
Algorithm~\ref{alg:classic-sia-conc} terminates with an optimal strategy $\sigma$ 
($\val_{\mathcal G}(\sigma) = \val_{\mathcal G}$) for \Pmax.
\end{theorem}
As a remark, we can drop the combinability requirement while maintaining
correctness when we restrict the updates to a single position, that is, when we
require $P$ to be singleton for every update. 
We call such strategy improvement algorithms \emph{slow}, and a class of games
\emph{good for slow $\max$ strategy improvement} if it is maximum identifying
and positionally determined. 

\begin{theorem}
\label{theorem:slow}
If a class of games is  positionally determined games with maximum identifying
improvement then  all slow strategy improvement algorithms terminate with an optimal
strategy $\sigma$ 
($\val_{\mathcal G}(\sigma) = \val_{\mathcal G}$) for \Pmax.
\end{theorem}

The proof for both theorems is the same.

\begin{proof}
The strategy improvement algorithm will produce a sequence $\sigma_0, \sigma_1,
\sigma_2 \ldots$ of positional strategies with increasing quality $\val_{\mathcal
  G}(\sigma_0) \sqsubset \val_{\mathcal G}(\sigma_1) \sqsubset \val_{\mathcal
  G}(\sigma_2) \sqsubset \ldots$. As the set of positional strategies is finite,
this chain must be finite. 
As the game is maximum identifying, the stopping condition provides optimality.
\end{proof}

Various concepts and results extend naturally for analogous claims about
\Pmin.
We call a class of game \emph{good for strategy improvement} if it is good for $\max$ 
strategy improvement and good for $\min$ strategy improvement. 
Parity games, mean payoff games, and discounted payoff games are all good for
strategy improvement (for both players).
Moreover, the calculation of $\prof(\sigma)$ is cheap in all of these instances,
which makes them well suited for strategy improvement techniques.

\section{Symmetric Strategy Improvement Algorithm}
\label{sec:algo}
We first extend the termination argument for classic strategy improvement
techniques (Theorems \ref{theorem:classic} and \ref{theorem:slow}) to symmetric
strategy improvement given as Algorithm~\ref{alg:classic-ssia-conc}.

\begin{algorithm}[h]
  \caption{\label{alg:classic-ssia-conc} Symmetric strategy improvement algorithm}
  Let $\sigma_0$ and $\tau_0$  be arbitrary positional strategies. {\bf set} $i :=0$.\\
  Determine $\sigma_{\tau_i}^c$ and $\tau_{\sigma_i}^c$\\
  $\sigma_{i+1} := {\sigma_i}^P$ for $P\subseteq \prof(\sigma) \cap \sigma^c_{\tau_i}$.\\ 
  $\tau_{i+1} := {\tau_i}^P$ for $P\subseteq  \prof(\tau) \cap \tau^c_{\sigma_i}$. \\
  {\bf if} $\sigma_{i+1} = \sigma_i$ and $\tau_{i+1} = \tau_i$ {\bf return} $(\sigma_i, \tau_i)$.\\
  {\bf set} $i := i + 1$. {\bf go to} 2.\\
\end{algorithm}

\subsection{Correctness}
\begin{lemma}
The symmetric strategy improvement algorithm terminates for all classes of games that are good for strategy improvement.
\end{lemma}

\begin{proof}
We first observe that the algorithm yields a sequence $\sigma_0, \sigma_1,
\sigma_2, \ldots$ of \Pmax\ strategies for $\mathcal G$ with improving values $\val_{\mathcal G}(\sigma_0) \sqsubseteq \val_{\mathcal G}(\sigma_1) \sqsubseteq \val_{\mathcal G}(\sigma_2) \sqsubseteq \ldots$, where equality, $\val_{\mathcal G}(\sigma_i) \equiv \val_{\mathcal G}(\sigma_{i+i})$, implies $\sigma_i = \sigma_{i+1}$.
Similarly, for the sequence $\tau_0, \tau_1, \tau_2, \ldots$ of \Pmin\ strategies for $\mathcal G$, the values $\val_{\mathcal G}(\tau_0) \sqsupseteq \val_{\mathcal G}(\tau_1) \sqsupseteq \val_{\mathcal G}(\tau_2) \sqsupseteq \ldots$, improve (for \Pmin), such that equality, $\val_{\mathcal G}(\tau_i) \equiv \val_{\mathcal G}(\tau_{i+i})$, implies $\tau_i = \tau_{i+1}$.
As the number of values that can be taken is finite, eventually both values stabilise and the algorithm terminates.
\end{proof}

What remains to be shown is that the symmetric strategy improvement algorithm cannot terminate with an incorrect result. In order to show this, we first prove the weaker claim that it is optimal in $\mathcal G(\sigma,\tau,\sigma^c_\tau,\tau^c_\sigma) = (V_{\max},V_{\min},E',\val)$ such that
$E' = \big\{ \big(v,\sigma(v)\big) \mid v \in V_{\max}\big\}
\cup \big\{ \big(v,\tau(v)\big) \mid v \in V_{\min}\big\}
\cup \big\{ \big(v,\sigma^c_\tau(v)\big) \mid v \in V_{\max}\big\}
\cup \big\{ \big(v,\tau^c_\sigma(v)\big) \mid v \in V_{\min}\big\}$ is the subgame of $\mathcal G$ whose edges are those defined by the four positional strategies,
when it terminates with the strategy pair $\sigma,\tau$.

\begin{lemma}
\label{lemma:local}
When the symmetric strategy improvement algorithm terminates with the strategy
pair $\sigma,\tau$ on games that are good for strategy improvement, then
$\sigma$ and $\tau$ are the optimal strategies for Players Max and Min, respectively, in $\mathcal G(\sigma,\tau,\sigma^c_\tau,\tau^c_\sigma)$.
\end{lemma}

\begin{proof}
For $\mathcal G(\sigma,\tau,\sigma^c_\tau,\tau^c_\sigma)$, both update steps are not restricted:
the changes \Pmax\ can potentially select his updates from are the edges defined by $\sigma^c_\tau$ at the vertices $v\in V_{\max}$ where $\sigma$ and $\sigma^c_\tau$ differ ($\sigma(v) \neq \sigma^c_\tau(v)$).
Consequently, $\prof(\sigma) = \prof(\sigma)\cap \sigma^c_\tau$.

Thus, $\sigma=\sigma'$ holds if, and only if, $\sigma$ is the result of an
update step when using classic strategy improvement in $\mathcal
G(\sigma,\tau,\sigma^c_\tau,\tau^c_\sigma)$ when starting in $\sigma$. As game
is maximum identifying, $\sigma$ is the optimal \Pmax\ strategy for $\mathcal G(\sigma,\tau,\sigma^c_\tau,\tau^c_\sigma)$.

Likewise, the \Pmin\ can potentially select every updates from $\tau^c_\sigma$, at vertices $v\in V_{\min}$ and we first get $\prof(\tau) = \prof(\tau)\cap \tau^c_\sigma$ with the same argument.
As the game is minimum identifying, $\tau$ is the optimal \Pmin\ strategy for $\mathcal G(\sigma,\tau,\sigma^c_\tau,\tau^c_\sigma)$.
\end{proof}

We are now in a position to expand the optimality in the subgame $\mathcal
G(\sigma,\tau,\sigma^c_\tau,\tau^c_\sigma)$ from Lemma \ref{lemma:local} to
global optimality the valuation of these strategies for $\mathcal G$. 

\begin{lemma}
\label{lemma:evaluation}
When the symmetric strategy improvement algorithm terminates with the strategy
pair $\sigma,\tau$ on a game $\mathcal G$ that is good for strategy improvement,
then $\sigma$ is an optimal \Pmax\ strategy and $\tau$ an optimal \Pmin\ strategy. 
\end{lemma}

\begin{proof}
Let $\sigma,\tau$ be the strategies returned by the symmetric strategy improvement algorithm for a game $\mathcal G$, and let $\mathcal L = \mathcal G(\sigma,\tau,\sigma^c_\tau,\tau^c_\sigma)$ denote the local game from Lemma \ref{lemma:local} defined by them.
Lemma \ref{lemma:local} has established optimality in $\mathcal L$.
Observing that the optimal responses in $\mathcal G$ to $\sigma$ and $\tau$, $\tau^c_\sigma$ and $\sigma^c_\tau$, respectively, are available in $\mathcal L$, we first see that they are also optimal in $\mathcal L$.
Thus, we have
\begin{itemize}
\item $\val_{\mathcal L}(\sigma) \equiv \val_{\mathcal L}(\sigma,\tau^c_\sigma) \equiv \val_{\mathcal G}(\sigma,\tau^c_\sigma)$ and
\item  $\val_{\mathcal L}(\tau) \equiv \val_{\mathcal L}(\sigma^c_\tau,\tau) \equiv \val_{\mathcal G}(\sigma^c_\tau,\tau)$.
\end{itemize}
Optimality in $\mathcal L$ then provides $\val_{\mathcal L}(\sigma) = \val_{\mathcal L}(\tau)$. Putting these three equations together, we get $\val_{\mathcal G}(\sigma,\tau^c_\sigma) \equiv \val_{\mathcal G}(\sigma^c_\tau,\tau)$.

Taking into account that $\tau^c_\sigma$ and $\sigma^c_\tau$ are the optimal responses to $\sigma$ and $\tau$, respectively, in $\mathcal G$, we expand this to $\val_{\mathcal G} \sqsupseteq \val_{\mathcal G}(\sigma) \equiv \val_{\mathcal G}(\sigma,\tau^c_\sigma) \equiv \val_{\mathcal G}(\sigma^c_\tau,\tau) \equiv \val_{\mathcal G}(\tau) \sqsupseteq \val_{\mathcal G}$ and get 
$\val_{\mathcal G} \equiv \val_{\mathcal G}(\sigma) \equiv \val_{\mathcal G}(\tau) \equiv \val_{\mathcal G}(\sigma,\tau)$.
\end{proof}

The Lemmas in this subsection yield the following results.
\begin{theorem}
\label{theorem:correct}
The symmetric strategy improvement algorithm is correct for games that are good for strategy improvement.
\end{theorem}
\begin{theorem}
\label{theorem:correctAndSlow}
The slow symmetric strategy improvement algorithm is correct for positionally determined games that are maximum and minimum identifying.
\end{theorem}

We implemented our symmetric strategy improvement algorithm based on the
progress measures introduced by V\"oge and
Jurdzi\'nski \cite{Voge+Jurdzinski/00/simprove}. 
The first step is to determine the valuation for the optimal counter strategies
to and the valuations for $\sigma$ and $\tau$. 

\begin{example} 
In our running example from Figure \ref{fig:example}, we have discussed in the
previous section that $\tau$ is the optimal counter strategy $\tau^c_\sigma$ and
that $\prof(\sigma) = \{(3,4),(3,0)\}$. 
In the optimal counter strategy $\sigma^c_\tau$ to $\tau$, \PZero{} moves from
$3$ to $4$, and we get $\val(1,\tau) = (1,\emptyset,0)$, $\val(4,\tau) =
(4,\emptyset,0)$, $\val(3,\tau) = (4,\emptyset,1)$, and $\val(0,\tau) =
(0,\emptyset,0)$. 
Consequently, $\prof(\tau) = \{(4,1)\}$.
For the update of $\sigma$, we select the intersection of $\prof(\sigma)$ and $\sigma^c_\tau$.
In our example, this is the edge from $3$ to $4$ (depicted in {\color{green!70!black} green}).
To update $\tau$, we select the intersection of $\prof(\tau)$ and $\tau^c_\sigma$.
In our example, this intersection is empty, as the current strategy $\tau$ agrees with $\tau^c_\sigma$.
\end{example}

\subsection{A minor improvement on stopping criteria} 
In this subsection, we look at a minor albeit natural improvement over
Algorithm~\ref{alg:classic-ssia-conc} shown in
Algorithm~\ref{alg:classic-ssia-conc-alt}. 
There we used termination on both sides as a condition to terminate the algorithm.
We could alternatively check if \emph{either} player has reached an optimum.
Once this is the case, we can return the optimal strategy and an optimal counter
strategy to it. 
\begin{algorithm}[h]
  \caption{\label{alg:classic-ssia-conc-alt} Symmetric strategy improvement
    algorithm (Improved Stopping criteria)}
  Let $\sigma_0$ and $\tau_0$  be arbitrary positional strategies. {\bf set} $i :=0$.\\
  Determine $\sigma_{\tau_i}^c$ and $\tau_{\sigma_i}^c$\\
  {\bf if}  $\prof(\sigma_i) = \emptyset$ {\bf return} $(\sigma_i, \tau_{\sigma_i}^c)$;\\
  {\bf if}  $\prof(\tau_i) = \emptyset$ {\bf return} $(\sigma_{\tau_i}^c, \tau_i)$;\\
  $\sigma_{i+1} := {\sigma_i}^P$ for $P\subseteq \prof(\sigma) \cap \sigma^c_{\tau_i}$.\\ 
  $\tau_{i+1} := {\tau_i}^P$ for $P\subseteq  \prof(\tau) \cap \tau^c_{\sigma_i}$. \\
  {\bf set} $i := i + 1$. {\bf go to} 2.\\
\end{algorithm}

The correctness of this stopping condition is provided by Theorems
\ref{theorem:classic} and \ref{theorem:slow}, and checking this stopping
condition is usually cheap: it suffices to check if $\prof(\sigma)$ or
$\prof(\tau)$ is empty. 
This provides us with a small optimisation, as we can stop as soon as one of
the strategies involved is optimal. 
However this small optimisation can only provide a small advantage.

\begin{theorem}
\label{theorem:optlin}
The difference in the number of iterations of
Algorithm~\ref{alg:classic-ssia-conc} and
Algorithm~\ref{alg:classic-ssia-conc-alt} is at most linear in the number of states of $\mathcal G$.
\end{theorem}

\begin{proof}
Let $\sigma$ be an optimal strategy for $\mathcal G$.
When starting with a strategy pair $\sigma,\tau_0$ for some strategy $\tau_0$ of
\Pmin, we first construct the optimal counter strategies $\tau^c_\sigma$ and $\sigma_{\tau_0}$.
As $\sigma$ is optimal and $\mathcal G$ maximum identifying,
$\prof(\sigma)=\emptyset$, and strategy improvement will not change it. 
In particular, our algorithm will always provide $\sigma' = \sigma$,
irrespective of the optimal counter strategy $\sigma_{\tau_i}^c$ to a strategy
$\tau_i$ of \Pmin.
This also implies that $\tau^c_\sigma$ will not change.
It is now easy to see that, unless $\tau_i' = \tau_i$, $\tau_{i+1} = \tau_i'$
differs from $\tau_i$ in at least one decision, and it differs by adhering to
$\tau^c_\sigma$ at the positions where it differs ($\forall v \in
V_{\min}.\ \tau_i(v) \neq \tau_{i+1}(v) \Rightarrow \tau_{i+1}(v) =
\tau^c_\sigma(v)$). 
Such an update can happen at most once for each \Pmin\ position.
 The argument for starting with an optimal strategy $\tau$ of \Pmin\ is similar. 
\end{proof}

\section{Friedmann's Traps}
\label{sec:friedmann}

In a seminal work on the complexity of strategy improvement
\cite{Friedmann/11/lower}, Friedmann uses a class of parity games called
\emph{1-sink parity games}. These games contain a \emph{sink} node with 
the weakest odd parity in a max-parity game. This sink node is reachable from
every other node in the game and such a game is won by \POne{} eventually. 
Figure \ref{fig:lbgame} shows a lower bound game from \cite{Friedmann/11/lower}.

In order to obtain an exponential lower bound for the classic strategy improvement algorithm with the locally optimising policy, these sink games implement a binary counter realised by a gadget called a \emph{cycle gate} which consists of two components. With $n$ cycle gates, we have a representation of the $n$ bits for an $n$ bit counter. The first component of a cycle gate is called a \emph{simple cycle}. In Figure \ref{fig:lbgame}, the three smaller boxes shown in yellow are the simple cycles of the game. These simple cycles encode the bits of the counter. The second component of the cycle gate gadget is called a \emph{deceleration lane}. This structure serves to ensure that any profitable updates to strategies are postponed by cycling through seemingly more profitable improvements, in the order $r, s, a_1, a_2, \ldots$, before eventually turning to $e_i$.
This structure is shown as a shaded blue rectangle in Figure \ref{fig:lbgame}.

\begin{figure*}
  \begin{center}
\scalebox{0.636}{
    \begin{tikzpicture}
      \tikzstyle{dummy}=[draw,dashed,circle,minimum size=2em,inner sep=0em]
      \tikzstyle{eloc}=[draw,circle,minimum size=3.5em,inner sep=0em]
      \tikzstyle{oloc}=[draw,minimum size=3em,inner sep=0em]

      \tikzstyle{trans}=[-latex, rounded corners]

      \node[eloc] at (0,0) (c)  {$c:28$};
      \node[eloc] at (0,2) (t1)  {$t_1:15$};
      \node[eloc] at (0,4) (t2)  {$t_2:17$};
      \node[eloc] at (0,6) (t3)  {$t_3:19$};
      \node[eloc] at (0,8) (t4)  {$t_4:21$};
      \node[eloc] at (0,10) (t5)  {$t_5:23$};
      \node[eloc] at (0,12) (t6)  {$t_6:25$};

      \node[oloc] at (2,2) (a1)  {$a_1:16$};
      \node[oloc] at (2,4) (a2)  {$a_2:18$};
      \node[oloc] at (2,6) (a3)  {$a_3:20$};
      \node[oloc] at (2,8) (a4)  {$a_4:22$};
      \node[oloc] at (2,10) (a5)  {$a_5:24$};
      \node[oloc] at (2,12) (a6)  {$a_6:26$};

      \fill[blue!20!white,draw, fill opacity=0.2] (-2, -1) rectangle (3, 14);

      \node[eloc] at (6,3) (d1)  {$d_1:3$};
      \node[eloc] at (6,7) (d2)  {$d_2:7$};
      \node[eloc] at (6,11) (d3)  {$d_3:11$};

      \node[oloc] at (8,3) (e1)  {$e_1:4$};
      \node[oloc] at (8,7) (e2)  {$e_2:8$};
      \node[oloc] at (8,11) (e3)  {$e_3:12$};

      \fill[yellow!40!white,draw, fill opacity=0.2] (5,2.2) rectangle (9, 3.8);
      \fill[yellow!40!white,draw, fill opacity=0.2] (5,6.2) rectangle (9, 7.8);
      \fill[yellow!40!white,draw, fill opacity=0.2] (5,10.2) rectangle (9, 11.8);

      \node[oloc] at (9,1) (f1)  {$f_1:35$};
      \node[oloc] at (9,5) (f2)  {$f_2:39$};
      \node[oloc] at (9,9) (f3)  {$f_3:43$};

      \node[oloc] at (13,3) (h1)  {$h_1:36$};
      \node[oloc] at (13,7) (h2)  {$h_2:40$};
      \node[oloc] at (13,11) (h3)  {$h_3:44$};

      \node[eloc] at (13,1) (g1)  {$g_1:6$};
      \node[eloc] at (13,5) (g2)  {$g_2:10$};
      \node[eloc] at (13,9) (g3)  {$g_3:14$};

      \node[eloc] at (16,3) (k1)  {$k_1:33$};
      \node[eloc] at (15.2,7) (k2)  {$k_2:37$};
      \node[eloc] at (16,11) (k3)  {$k_3:41$};

      \node[eloc, fill=black!20!white] at (9,-1) (r)  {$r:32$};
      \node[oloc] at (17,7) (x)  {$x:1$};
      \node[eloc,fill=black!20!white] at (11.5,13) (s)  {$s:30$};

      \draw[trans] (a1) -- (t1);
      \draw[trans] (a2) -- (t2);
      \draw[trans] (a3) -- (t3);
      \draw[trans] (a4) -- (t4);
      \draw[trans] (a5) -- (t5);
      \draw[trans] (a6) -- (t6);

      \draw[trans] (d1) -- (a1);
      \draw[trans] (d1) -- (a2);

      \draw[trans] (d2) -- (a1);
      \draw[trans] (d2) -- (a2);
      \draw[trans] (d2) -- (a3);
      \draw[trans] (d2) -- (a4);

      \draw[trans] (d3) -- (a1);
      \draw[trans] (d3) -- (a2);
      \draw[trans] (d3) -- (a3);
      \draw[trans] (d3) -- (a4);
      \draw[trans] (d3) -- (a5);
      \draw[trans] (d3) -- (a6);

      \draw[trans] (t6) -- (t5);
      \draw[trans] (t5) -- (t4);
      \draw[trans] (t4) -- (t3);
      \draw[trans] (t3) -- (t2);
      \draw[trans] (t2) -- (t1);
      \draw[trans] (t1) -- (c);

      \draw[trans] (t6) --  +(-1, +0.5) node[dummy,pos=1,left] {$r$};
      \draw[trans] (t6) --  +(-1, -0.5) node[dummy,pos=1,left] {$s$};

      \draw[trans] (t5) --  +(-1, +0.5) node[dummy,pos=1,left] {$r$};
      \draw[trans] (t5) --  +(-1, -0.5) node[dummy,pos=1,left] {$s$};

      \draw[trans] (t4) --  +(-1, +0.5) node[dummy,pos=1,left] {$r$};
      \draw[trans] (t4) --  +(-1, -0.5) node[dummy,pos=1,left] {$s$};

      \draw[trans] (t3) --  +(-1, +0.5) node[dummy,pos=1,left] {$r$};
      \draw[trans] (t3) --  +(-1, -0.5) node[dummy,pos=1,left] {$s$};

      \draw[trans] (t2) --  +(-1, +0.5) node[dummy,pos=1,left] {$r$};
      \draw[trans] (t2) --  +(-1, -0.5) node[dummy,pos=1,left] {$s$};

      \draw[trans] (t1) --  +(-1, +0.5) node[dummy,pos=1,left] {$r$};
      \draw[trans] (t1) --  +(-1, -0.5) node[dummy,pos=1,left] {$s$};

      \draw[trans] (c) --  +(-1, +0.5) node[dummy,pos=1,left] {$r$};
      \draw[trans] (c) --  +(-1, -0.5) node[dummy,pos=1,left] {$s$};

      \draw[trans] (d1) --  +(-0.5, +1) node[dummy,pos=1,above] {$r$};
      \draw[trans] (d1) --  +(0.5, +1) node[dummy,pos=1,above] {$s$};

      \draw[trans] (d2) --  +(-0.5, +1) node[dummy,pos=1,above] {$r$};
      \draw[trans] (d2) --  +(0.5, +1) node[dummy,pos=1,above] {$s$};

      \draw[trans] (d3) --  +(-0.5, +1) node[dummy,pos=1,above] {$r$};
      \draw[trans] (d3) --  +(0.5, +1) node[dummy,pos=1,above] {$s$};

      \draw[trans] (d3) to [bend left=20] (e3);
      \draw[trans] (e3) to [bend left=20] (d3);

      \draw[trans] (d2) to [bend left=20] (e2);
      \draw[trans] (e2) to [bend left=20] (d2);

      \draw[trans] (d1) to [bend left=20] (e1);
      \draw[trans] (e1) to [bend left=20] (d1);

      \draw[trans] (g3) -- (f3);
      \draw[trans] (f3) -- (e3);
      \draw[trans] (e3) -- (h3);
      \draw[trans] (h3) -- (k3);
      \draw[trans] (g3) -- (k3);
      \draw[trans] (k3) -- (x);

      \draw[trans] (g2) -- (f2);
      \draw[trans] (f2) -- (e2);
      \draw[trans] (e2) -- (h2);
      \draw[trans] (h2) -- (k2);
      \draw[trans] (g2) -- (k2);
      \draw[trans] (k2) -- (x);
      \draw[trans] (k2) -- (g3);

      \draw[trans] (g1) -- (f1);
      \draw[trans] (f1) -- (e1);
      \draw[trans] (e1) -- (h1);
      \draw[trans] (h1) -- (k1);
      \draw[trans] (g1) -- (k1);
      \draw[trans] (k1) -- (x);
      \draw[trans] (k1) -- (g2);
      \draw[trans] (k1) -- (g3);

      \draw[trans] (s) -- (f1);
      \draw[trans] (s) -- (f2);
      \draw[trans] (s) -- (f3);
      \draw[trans] (s) -- +(5, 0) -| (x);

      \draw[trans] (r) -- (g1);
      \draw[trans] (r) -- (g2);
      \draw[trans] (r) -- (g3);
      \draw[trans] (r) -- +(5, 0) -| (x);

      \draw[trans] (x) to [loop right]  (x);

    \end{tikzpicture}
}
  \end{center}
  \caption{Friedmann's lower bound game for the locally optimal strategy improvement algorithm}
  \label{fig:lbgame}
\end{figure*}
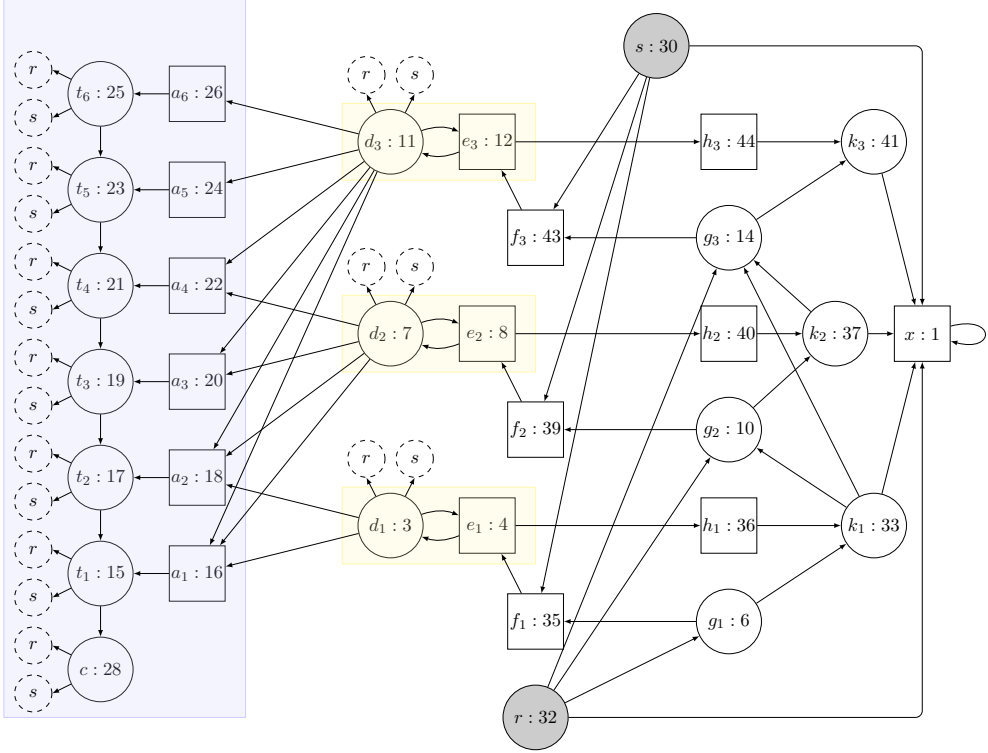

A simple cycle consists of exactly one \PZero{} controlled node $d$ with a weak
odd colour $k$ and one \POne{} controlled node $e$ with the even colour
$k+1$. The \PZero{} node is also connected to some set of external nodes in the
game  and the \POne{} node is connected to an output node with a high even
colour on a path to the sink node. Given a strategy $\sigma$, we say that a
simple cycle is closed if we have an edge $\sigma(d) =e$. 
Otherwise, we say that the simple cycle is open. Opening and closing cycles
correspond to unsetting and setting bits. We then say a cycle gate is
\emph{open} or \emph{closed} when its corresponding simple cycle is open or
closed respectively.  

In these lower bound games, the simple cycles are connected to the deceleration lane in such a way that lower valued cycles have less edges entering the deceleration lane ensuring that lower open cycles close before higher open cycles. This allows the lesser significant bits to be set and reset before the higher significant bits.

The deceleration lane hides sensible improvements, thus making the players take more iterations before taking the best improvement. 
It is then shown in \cite{Friedmann/11/lower} that incrementing a bit state always requires more than one strategy iteration in 4 different phases.
This gadget thus counts an exponential number of improvement steps taken by the strategy improvement algorithm to flip $n$ bits. 
For a detailed exposition of the gadget and the exponential lower bound
construction, we refer the reader to \cite{Friedmann/11/lower}.

\subsection{Escaping the traps with symmetric strategy improvement}
We discuss the effect of symmetric strategy improvement on Friedmann's traps,
with a focus on the simple cycles. 
Simple cycles are the central component of the cycle gates and the heart of the
lower bound proof. 
As described above, an $n$-bit counter is represented by $n$ cycle gates, each
cycle gate embedding a smaller simple cycle. 
These simple cycles are reused exponentially often to represent $n$ bits. Both
players have the choice to open or close the simple cycles. 

The optimal strategy of both players in the simple cycles of Figure
\ref{fig:lbgame} is to turn right. 
(For \PZero{}, one could say that he wants to leave the cycle, and for \POne{},
one could say that she wants to stay in it.) 
When the players agree to stay in the cycle, \PZero{} wins the parity game. In
fact these are the only places where \PZero{} can win positionally in this parity game. 
When running the symmetric strategy improvement algorithm for \PZero{}, the
optimal counter strategy by \POne{} is to move to the right in simple cycles
where \PZero{} is moving to the right, and to move left in all other simple
cycles. 

As mentioned before, Friedmann \cite{Friedmann/11/lower} showed that, when
looking at an abstraction of the \PZero{} strategy that only distinguishes the
decisions of turning right or not turning right in the simple cycles, then they
essentially behave like a binary counter that, with some delay (caused by the
deceleration lane) will `count up'. More precisely, one step after the $i^{th}$
bit has been activated, all lower bits are reset. 

We now discuss how symmetric strategy improvement can beat this mechanism by
taking the view of both players into account. 
For this, we consider a starting configuration, where \POne{} moves to the right
in the $j$ most significant simple cycle positions, where $j$ can be $0$. 
Note that, when \POne{} moves right in all of these positions, she has found her
optimal strategy and we can invoke Theorem \ref{theorem:optlin} to show that the
algorithm terminates in a linear number of steps---or simply stop when using the
alternative stopping condition. 

The first observation is that changing the decision to moving left will not lead
to an improvement, as it produces a winning cycle of a quality (leading even
colour) higher than the quality of any cycle available for \PZero{} under the
current strategy of \POne{}. 
Let us now consider the less significant position $j+1$. 
First, we observe that moving to the right is a superior strategy.
This can easily be seen:
moving to the left produces a cycle with a dominating even colour and thus turns
out to be winning for \PZero{}.  
Moving to the right in position $j+1$ and (by our assumption) all more
significant positions removes this cycle and implies that the leading colour
from this position is 1. This is clearly better for \POne{}. 
If \POne{} uses a strategy where $j+1$ is the most significant position where
she decides to move to the left, we have the following case distinctions for
\PZero{}'s strategy in this simple cycle: 
\begin{enumerate}
\item 
\PZero{} moves to the right in this simple cycle. Then moving to the right is
also the optimal counter strategy for \POne{}, and her strategy will be updated
accordingly. 
\item 
  \PZero{} does not move right in this simple cycle with her current strategy $\sigma$. 
  Moving right in this simple cycle is among $\prof(\sigma)$, as one even colour
  is added to the set in the quality measure in the local comparison. It is also
  the choice for the optimal counter strategy $\sigma^c_\tau$ to the current
  strategy $\tau$ of \POne{}, as this is the only way for \PZero{} to produce a
  valuation with the dominating even colour of this simple cycle, while to
  valuation with a higher even colour is possible. 
\end{enumerate}
Taking these two cases into consideration, \POne{} will move to the right in the
$j$ most significant positions after $2j$ improvement steps. 
When \PZero{} has found his optimal strategy, we can invoke Theorem
\ref{theorem:optlin} to show termination in linear steps for the algorithm. 

There are similar arguments for all kinds of traps that Friedmann has developed
for strategy improvement algorithms. We have not formalised these arguments on
other instances, but provided the number of iterations needed by our symmetric
strategy improvement algorithm for all of them in the next section. 

Note that the way in which Friedmann traps asymmetric strategy improvement has proven to be quite resistant to the improvement policy (snare \cite{Fearnley/10/snare}, random facet \cite{Ludwig/95/random,BjorklundVorobyov/07/subexp}, globally optimal \cite{Schewe/08/improvement}, etc.).
From the perspective of the traps, the different policies try to aim at a minor point in the mechanism of the traps, and this minor point is adjusted. The central mechanism, however is not affected.
All of these examples have some variant of simple cycles at the heart of the counter and a deceleration lane to orchestrate the timely counting.

Symmetric strategy improvement aims at the mechanism of the traps themselves.
It seems that examples that trap symmetric strategy improvement algorithms need to do more than just trapping both players (which could be done by copying the trap with inverse roles), they need to trap them simultaneously.
It is not likely to find a proof that such traps do not exist, as this would imply a proof that symmetric strategy improvement solves parity (or, depending on the proof, mean or discounted payoff) games in polynomial time.
But it seems that such traps would need a different structure.
A further difference to asymmetric strategy improvement is that the deceleration lane ceases to work.

Taking into account that finding traps for asymmetric strategy improvement took decades and was very insightful, this looks like an interesting challenge for future research.

\section{Experimental Results}
\label{sec:expt}
We have implemented the symmetric strategy improvement algorithm for parity games and
compared it with the standard strategy improvement algorithm with the popular
locally optimising and other switching rules.
To generate various examples we used the tools \texttt{steadygame} and
\texttt{stratimprgen} that comes as a part of the parity game solver collection
\textsc{PGSolver}~\cite{LF09}.
We have compared the performance of our algorith on parity games with 100 positions
(see appendix) and found that the locally optimising policy outperforms other 
switching rules. We therefore compare our symmetric strategy improvement algorithm
with the locally optimising strategy improvement below.

Since every iteration of both algorithms is rather similar---one iteration of
our symmetric strategy improvement algorithm essentially runs two copies of an
iteration of a classical strategy improvement algorithm---and can be performed
in polynomial time, the key data to compare these algorithms is the number of
iterations taken by both algorithms.

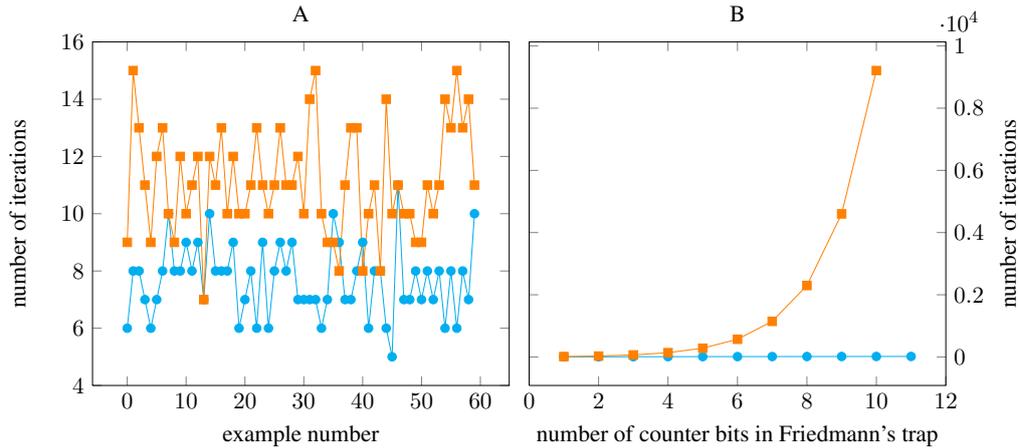
\begin{figure}[h]
\scalebox{0.8}{
  \begin{tikzpicture}
    \begin{axis}[
        title={A}, 
        xlabel={example number},
        ylabel={number of iterations},
        legend pos=north west,
        legend style={fill=none}       
      ]
      \addplot[
        color=cyan,
        mark=*,
       ]
      coordinates {
        (0,6)(1,8)(2,8)(3,7)(4,6)(5,7)(6,8)(7,10)(8,8)(9,8)
        (10,9)(11,8)(12,9)(13,7)(14,10)(15,8)(16,8)(17,8)(18,9)(19,6)
        (20,7)(21,8)(22,6)(23,9)(24,6)(25,8)(26,9)(27,8)(28,9)(29,7)
        (30,7)(31,7)(32,7)(33,6)(34,7)(35,10)(36,9)(37,7)(38,7)(39,8)
        (40,9)(41,6)(42,8)(43,8)(44,6)(45,5)(46,11)(47,7)(48,7)(49,8)
        (50,7)(51,8)(52,7)(53,8)(54,6)(55,8)(56,6)(57,8)(58,7)(59,10)
      };
      \addplot[
        color=orange,
        mark=square*]
      coordinates {
        (0,9)(1,15)(2,13)(3,11)(4,9)(5,12)(6,13)(7,10)(8,9)(9,12)
        (10,10)(11,11)(12,12)(13,7)(14,12)(15,11)(16,13)(17,10)(18,12)(19,10)
        (20,10)(21,11)(22,13)(23,11)(24,10)(25,11)(26,13)(27,11)(28,11)(29,12)
        (30,10)(31,14)(32,15)(33,10)(34,9)(35,9)(36,8)(37,11)(38,13)(39,13)
        (40,8)(41,10)(42,11)(43,8)(44,14)(45,10)(46,11)(47,10)(48,10)(49,9)
        (50,9)(51,11)(52,10)(53,11)(54,14)(55,13)(56,15)(57,13)(58,14)(59,11)
      };
     \end{axis}
  \end{tikzpicture}
  \hfill
  \begin{tikzpicture}
    \begin{axis}[
        title={B},
        xlabel={number of counter bits in Friedmann's trap},
        ylabel={number of  iterations},
        yticklabel pos=right,ylabel near ticks,
        legend pos=north west,
      ]
      \addplot[
        color=cyan,
        mark=*,
      ]
      coordinates {
        (1,3)(2,5)(3,7)(4,9)(5,10)(6,11)(7,12)(8,13)(9,14)(10,15)(11,16)
      };
      \addplot[
        color=orange,
        mark=square*
      ]
      coordinates {
        (1,11)(2,29)(3,65)(4,137)(5,281)(6,569)(7,1145)(8,2297)(9,4601)(10,9209)
      };
     \end{axis}
  \end{tikzpicture}
}
\caption{These plots compare the performance of the symmetric
  strategy improvement algorithm  (data points in cyan circles) with standard strategy 
  improvement using the locally optimising policy rule (data points in orange
  squares). 
  The plot on the left side is for random examples generated using the 
  \texttt{steadygame 1000 2 4 3 5 6} command, while the plot on the right is for
  Friedmann's trap from the previous section generated by the command 
  \texttt{stratimprgen -pg switchallsubexp i}.}
\label{fig:expfigs}
\end{figure}

Symmetric strategy improvement will often rule out improvements at individual
positions: it disregards profitable changes of Player Max and Min if they do not
comply with $\sigma^c_\tau$ and $\tau^c_\sigma$, respectively. 
It is well known that considering fewer updates can lead to a significant
increase in the number of updates on random examples and benchmarks. 
An algorithm based on the random-facet
method \cite{Ludwig/95/random,BjorklundVorobyov/07/subexp}, e.g., needs around a
hundred iterations on the random examples with 100 positions we have drawn, simply
because it updates only a single position at a time. 
The same holds for a random-edge policy where only a single position is updated.
The figures for these two methods are given in the appendix.

It is therefore good news that symmetric strategy improvement does not display a
similar weakness. 
It even uses less updates when compared to classic strategy improvement with the
popular locally optimising and locally random policy rules. 
Note also that having less updates can lead to a faster evaluation of the
update, because unchanged parts do not need to be
re-evaluated~\cite{BjorklundVorobyov/07/subexp}. 

As shown in Figure~\ref{fig:expfigs}, the symmetric strategy improvement
algorithm not only performs better (on average) in comparison with the traditional
strategy improvement algorithm with the locally optimising policy rule, but also
avoids Friedmann's traps for the strategy improvement algorithm. 
The following table shows the performance of symmetric strategy improvement
algorithm for Friedmann's traps for other common switching rules.
It is clear that our algorithm is not exponential for these classes of
examples. 
\begin{center}
\begin{tabular}{l | c c c c c c c c c c}
\hline
Switch Rule & 1 & 2 & 3 & 4 & 5 & 6 & 7 & 8 & 9 & 10 \\ [0.5ex] 
\hline
Cunningham & 2 &6&9&12&15&18&21&24&27&30\\
CunninghamSubexp &1&1&1&1&1&1&1&1&1&1\\
FearnleySubexp &4&7&11&13&17&21&25&29&33&37\\
FriedmannSubexp&4&9&13&15&19&23&27&31&35&39\\
RandomEdgeExpTest &1&2&2&2&2&2&2&2&2&2\\
RandomFacetSubexp &1&2&7&9&11&13&15&17&19&21\\
SwitchAllBestExp &4&5&8&11&12&13&15&17&18&19\\
SwitchAllBestSubExp &5&7&9&11&13&15&17&19&21&23\\
SwitchAllSubExp &3&5&7&9&10&11&12&13&14&15\\
SwitchAllExp &3&4&6&8&10&11&12&14&16&18\\
ZadehExp &-&6&10&14&18&21&25&28&32&35\\
ZadehSubexp &5&9&13&16&20&23&27&30&34&37\\
\hline
\end{tabular}
\end{center}

\section{Discussion}
\label{sec:discuss}
We have introduced symmetric approaches to strategy improvement, where the
players take inspiration from the respective other's strategy when improving
theirs. 
This creates a rather moderate overhead, where each step is at most twice as expensive
as a normal improvement step. 
For this moderate price, we have shown that we can break the traps Friedmann has
introduced to establish exponential bounds for the different update policies in
classic strategy improvement
\cite{Friedmann/11/lower,Friedmann/11/Zadeh,Friedmann/13/snare}. 

In hindsight, attacking a symmetric problem with a symmetric approach seems so
natural, that it is quite surprising that it has not been attempted
immediately. 
There are, however, good reasons for this, but one should also consent that the
claim is not entirely true: 
the concurrent update to the respective optimal counter strategy has been
considered quite early
\cite{Friedmann/11/lower,Friedmann/11/Zadeh,Friedmann/13/snare}, but was
dismissed, because it can lead to cycles \cite{Condon93onalgorithms}.

The first reason is therefore that it was folklore that symmetric strategy
improvement does not work. 
The second reason is that the argument for the techniques that we have developed
in this paper would have been restricted to beauty until some of
the appeal of classic strategy improvement was caught in Friedmann's
traps. Friedmann himself, however, remained optimistic: 

\begin{quote}
We think that the strategy iteration still is a promising candidate for a
polynomial time algorithm, however it may be necessary to alter more of it than
just the improvement policy.
\end{quote}

\noindent This is precisely, what the introduction of symmetry and co-improvement tries to do.

\bibliography{bib}
 \newpage
 \appendix
\section{Symmetric Strategy Improvement algorithm  versus classic
  strategy improvement algorithm with various switching rules}

\begin{figure}[h]
\begin{center}
  \begin{tikzpicture}
    \begin{axis}[
        title={A}, 
        xlabel={example number},
        ylabel={number of iterations},
        legend pos=north west,
        legend style={fill=none}       
      ]
      \addplot[
        color=cyan,
        mark=*,
       ]
      coordinates {
        (1,5)(2,4)(3,5)(4,7)(5,6)(6,7)(7,6)(8,6)(9,5)(10,7)
      };
      \addplot[
        color=orange,
        mark=square*]
      coordinates {
        (1,5)(2,6)(3,8)(4,8)(5,10)(6,6)(7,8)(8,8)(9,5)(10,6)
      };

      \addplot[
        color=red,
        mark=triangle*]
      coordinates {
        (1,72)(2,60)(3,89)(4,113)(5,85)(6,79)(7,55)(8,96)(9,87)(10,64)
      };

      \addplot[
        color=blue,
        mark=triangle*]
      coordinates {
        (1,92)(2,63)(3,107)(4,88)(5,98)(6,96)(7,60)(8,93)(9,85)(10,63)
      };

      \addplot[
        color=green,
         mark=triangle*]
      coordinates {
        (1,15)(2,13)(3,13)(4,15)(5,14)(6,14)(7,14)(8,16)(9,16)(10,10)
       };
     \end{axis}
  \end{tikzpicture}
\end{center}

\caption{These plots compare the performance of the symmetric
  strategy improvement algorithm  (data points in cyan circles) with standard strategy 
  improvement using the locally optimising policy rule (data points in orange
  squares), random-edge switching rule (data points in red triangles),
  random-facet rule (data points in blue triangles), and switch-half rule (data
  point in green triangles). 
  These plots are for random examples generated using the 
  \texttt{steadygame 100 2 4 3 5 6} command from PGSolver. 
  The results from randomized switching rules (random-edge, random-facet, and
  switch-half)  presented here are taken as average number of iterations over
  four  executions of the corresponding algorithms.} 
\label{fig:expfigees}
\end{figure}
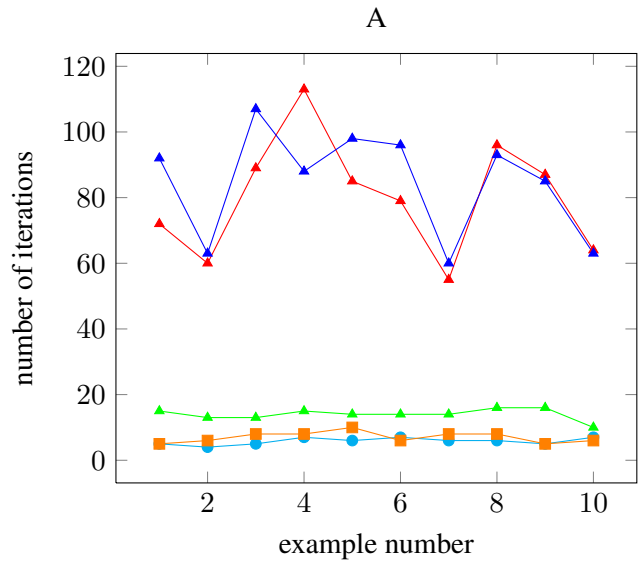

\end{document}